\documentclass[a4paper, 12pt]{article}

\usepackage[sort&compress]{natbib}
\bibpunct{(}{)}{;}{a}{}{,} 
\RequirePackage[colorlinks,citecolor=blue,urlcolor=blue]{hyperref}

\usepackage{amsthm, amsmath, amssymb, mathrsfs, multirow, url, subfigure}
\usepackage{graphicx} 
\usepackage{ifthen} 
\usepackage{amsfonts}
\usepackage[usenames]{color}
\usepackage{fullpage}
\usepackage{tikz}
\usepackage{float}
\usepackage{bbm}
\usepackage{booktabs}
\usepackage{soul}
\allowdisplaybreaks

\theoremstyle{plain} 
\newtheorem{thm}{Theorem}
\newtheorem{cor}{Corollary}

\newtheorem{lem}{Lemma}

\theoremstyle{definition}

\theoremstyle{remark}

\newcommand{\prob}{\mathsf{P}}

\newcommand{\unif}{{\sf Unif}}
\newcommand{\nm}{{\sf N}}

\newcommand{\geom}{{\sf Geom}}
\newcommand{\Tulap}{{\sf Tulap}}

\newcommand{\ZZ}{\mathbb{Z}}

\newcommand{\TT}{\mathbb{T}}

\newcommand{\iid}{\overset{\text{\tiny iid}}{\,\sim\,}}

\newcommand{\lPi}{\underline{\Pi}}
\newcommand{\uPi}{\overline{\Pi}}

\begin{document}

\title{\bf 
Prior- and likelihood-free probabilistic inference with finite-sample calibration guarantees
}
\author{Leonardo Cella\textsuperscript{$\star$} and Emily C. Hector\textsuperscript{$\dagger$}\hspace{.2cm}\\
\textsuperscript{$\star$}Department of Statistical Sciences, Wake Forest University\\
\textsuperscript{$\dagger$}Department of Biostatistics, University of Michigan}
\date{}
\maketitle
 
%\bigskip
\begin{abstract} %max 200 words
Motivated by parametric models for which the likelihood is analytically unavailable, numerically unstable, or prohibitively expensive to compute or optimize, we develop a prior- and likelihood-free framework for fully probabilistic (Bayesian-like) uncertainty quantification with finite-sample calibration guarantees. Our method, a type of inferential model, produces data-dependent degrees of belief about claims concerning the unknown parameter while controlling the frequency with which high belief is assigned to false claims, even in finite-sample settings. Our procedure is general in that it requires only the ability to simulate from the model. We first rank candidate parameter values according to how well data simulated from the model agree with the observed data, and then rescale these rankings in a way that yields the desired finite-sample calibration guarantees. The key idea is to employ a permutation-invariant function, such as a depth function, to rank parameter values. We show that such a choice yields closed-form calibration rescaling calculations, making the procedure computationally simple. We illustrate our method's broad appeal with four examples, including differential privacy and Ising models. An analysis of the spatial configuration of 2025 measles outbreaks in the U.S. showcases our method's practical advantages.
\end{abstract}

\noindent%
{\it Keywords:} Inferential models, exchangeability, Bayesian inference, simulation-based inference, finite-sample validity. %3-5 keywords

\section{Introduction}
\label{s:intro}

Modern statistical practice increasingly relies on complex models for which the likelihood is analytically unavailable, numerically unstable, or prohibitively expensive to compute or optimize. Simulation-based approaches, such as approximate Bayesian computation (ABC) \citep{beaumont2002approximate, fearnhead2012constructing, sisson2018overview}, neural Bayes \citep{lenzi2023neural, hector2024whole, zammit2024neural, sainsbury2024likelihood} and repro-sample frameworks \citep{xie2022repro,Awan03012025} bypass direct likelihood evaluation by relying on model-based simulation to compare observed and synthetic data. Approximation-based methods, including pseudo-likelihood and composite likelihood \citep{VarinReidFirthCompositeOverview, lindsay1988composite,besag1974}, replace the intractable likelihood with a tractable surrogate objective. Synthetic likelihood methods \citep{wood2010statistical,Price02012018} lie at the intersection of these two categories, using model-based simulation to construct parametric approximations to the likelihood of informative summaries. These methodological approaches cut across the traditional Bayesian–frequentist divide. Frequentist procedures provide inference through confidence sets and hypothesis tests with at least asymptotic error-rate control, and in some cases finite-sample guarantees, as in the repro-sample framework of \citet{Awan03012025}. Bayesian procedures, by contrast, yield fully probabilistic uncertainty quantification, assigning posterior probabilities—interpreted as degrees of belief—to claims about the unknown parameter, but come with no frequentist calibration guarantees, e.g., that the posterior probabilities assigned to false claims don't tend to be too large. To unify these two methodological strands, we develop a prior- and likelihood-free framework for fully probabilistic (Bayesian-like), finite-sample calibrated uncertainty quantification. 

Generally speaking, the search for a unification of frequentism and Bayesianism—what Efron calls ``the most important unresolved problem in statistical inference'' \citep{Efron2013}—has been the goal of several alternative frameworks, e.g., Fisher's fiducial approach \citep{fisherfiducial}, generalized fiducial inference \citep{MainHaning}, and confidence distributions \citep{mainconfdist}. The aim of this unification is to produce data-dependent degrees of belief, retaining the probabilistic nature of Bayesian inference, while ensuring that these beliefs are prior-free and calibrated, consistent with the frequentist view. The {\em false confidence theorem} \citep{Ryansatellite}, however, shows that belief assignments based on data-dependent probabilities lack calibration guarantees in that they are at risk of being systematically high for false claims. Consequently, calibrated belief assignments must take the form of suitable {\em non-additive} or {\em imprecise probabilities}. Among existing proposals—such as Dempster–Shafer theory \citep{dempster.copss,dempster1967,dempster1968a,shafer1976} and belief functions \citep{denoeux.li.2018,denoeux2014}—Inferential Models (IMs) \citep{martinbook,imreview} are the only framework that provides a construction guaranteeing belief calibration, and are therefore a serious candidate for addressing Efron’s challenge.

IMs are built on a simple two-step construction, reviewed in Section~\ref{s:IMs}: ranking and validification. This construction is quite general and can be applied in parametric, semiparametric, nonparametric, and prediction settings \citep[see][]{imreview,CellaBoots,CELLA2024IJAR,cella2025,CellaMartinSevere}. In the ranking step, a meaningful measure of compatibility between the observed data and the unknown quantity of interest is chosen; in the validification step, it is rescaled so that it leads to the type of calibration discussed above. Our focus here is on parametric settings, where the standard approach uses the likelihood in the ranking step. Besides being naturally principled, this construction inherits the efficiency typically associated with likelihood-based methods. This approach is not feasible, however, when the likelihood is unavailable or computationally intractable, particularly because validification often requires repeated likelihood evaluations.

A natural route to likelihood-free IMs is to replace the likelihood-based ranking with an alternative measure of data–parameter compatibility. While problem-specific alternatives have been proposed—for example, Gaussian ``working likelihoods'' in divide-and-conquer settings \citep{HectorCellaMartin2025}—we develop a general, principled, and computationally simple likelihood-free IM construction that retains the ranking–validification blueprint of IMs but implements the ranking step in a fundamentally different, simulation-based spirit. Our work is related to the distribution-free, prediction-focused IM of \citet{CellaMartin2021,CELLAMARTINREGPREDIJAR}, where exchangeability-preserving functions play a central role in achieving calibration through reasoning analogous to conformal prediction \citep{Vovk:2005}. In the present setting, exchangeability-preserving functions remain fundamental, but are now employed for parametric inference rather than distribution-free prediction. Our work is also closely connected to the ideas of \citet{Awan03012025}, but we move beyond their focus on procedures with guaranteed error-rate control, although such procedures can be derived from our approach. More specifically, our proposed construction is:
\begin{itemize}\setlength\itemsep{0em}
\item {\em general}, in that it applies whenever simulation from the assumed model is fast;
\item {\em principled}, in that parameter values are deemed plausible when they generate synthetic data that resemble the observed data, which can be interpreted as evidence that these values are supported by the data; and
\item {\em computationally simple}, combining depth-based ranking with closed-form validification.
\end{itemize}
Details of the proposed construction are presented in Section~\ref{s:lfIMs}, together with its finite-sample calibration guarantees. Finite-sample calibration is a central feature of the IM framework. Beyond achieving this property, the proposed construction introduces an important novelty: the resulting IM explicitly accounts for uncertainty arising from the simulated data used in the ranking step. Illustrations of the proposed method using synthetic data are presented in Section~\ref{s:examples}. A real-data application analyzing spatial measles incidence is given in Section~\ref{s:real}.

\section{Likelihood-based IMs: construction, limitations, and desiderata}

\subsection{Problem setup and motivation}
\label{s:IMs}

Let $Z \in \ZZ$ denote a random quantity of interest that is modeled by a probability distribution $\prob_\theta$ supported on $\ZZ$ and indexed by a finite-dimensional parameter $\theta \in \TT$. We assume there is a ``true'' parameter, denoted by $\Theta$, so that $Z\sim\prob_\Theta$. The goal is to quantify uncertainty about $\Theta$ after observing data $Z^n=z^n$, where $Z^n = Z_1, \ldots, Z_n$ consists of $n$ iid realizations from $\prob_\Theta$. We denote by $\prob_{\theta}^n$ the joint distribution of $Z^n$ under the model $Z_i \sim \prob_\theta$, $i=1, \ldots, n$.

The Bayesian–frequentist unification mentioned in Section~\ref{s:intro} seeks to produce prior-free, data-driven belief assignments about $\Theta$ with frequentist calibration guarantees. Specifically, if $bel_{z^n}(C)$ denotes the degree of belief assigned to the claim $\Theta \in C$ after observing data $z^n$, calibration requires that
\begin{equation}\label{eq:NeceVal}
\sup_{\theta \notin C}\prob_\theta^n \left\{ bel_{Z^n}(C) \geq 1-\alpha \right \} \leq \alpha, \quad  \text{for all $\alpha \in [0,1]$ and all $C \subseteq \TT$},
\end{equation}
so that high belief is rarely assigned to false claims. The false confidence theorem \citep{Ryansatellite} shows that beliefs modeled by additive probabilities cannot, in general, satisfy this requirement, implying the need for non-additive representations. However, non-additivity alone does not ensure \eqref{eq:NeceVal}; a special construction is required and, to our knowledge, IMs provide the only framework that guarantees it.

Let $\lPi_{z^n}$ represent the data-dependent IM degree of belief (or lower probability). Under non-additivity, there is a corresponding degree of plausibility (or upper probability), $\uPi_{z^n}$, such that $\lPi_{z^n} \le \uPi_{z^n}$, and the two are linked by the relation $\lPi_{z^n}(C) = 1 - \uPi_{z^n}(C^c)$. For IMs, they take the form of possibility and necessity measures \citep[e.g.,][]{dubois.prade.book, hose2021universal}, the simplest/more structured non-additive probability model. This implies the existence of a {\em possibility contour}, a function $\pi_{z^n}(\theta)$ on $\mathbb{Z}^{n} \times \TT \rightarrow [0,1]$ satisfying 
$\sup_{\theta \in \TT} \pi_{z^n}(\theta)=1$ for all $z^n$, from which the degree of plausibility is computed via
$\uPi_{z^n}(C) = \sup_{\theta \in C} \pi_{z^n}(\theta)$.

The IM possibility contour satisfies a key property that underlies all of its desirable features. Specifically, when viewed as a random variable, the IM contour is stochastically no smaller than a $\unif(0,1)$ random variable. This is known as the IMs' {\em validity property}, and in the present context, it takes the form
\begin{equation}\label{eq:IMvalidity}
\prob^n_{\Theta}\{ \pi_{Z^n}(\Theta) \leq \alpha\} \leq \alpha, \quad \text{for all $\alpha \in [0,1]$}.
\end{equation}
The likelihood-free IM proposed in Section~\ref{s:lfIMs} is designed for the same parametric context considered here, but its validity is achieved by incorporating additional random elements beyond the observed data $Z^n$ arising from the simulated data.

Before describing an IMs construction that guarantees \eqref{eq:IMvalidity}, we summarize in the following corollary the key features of IMs that follow from this property.

\begin{cor}\citep{CellaMartinMainSeverity}\label{cor:Main}
    For an IM whose contour satisfies the validity property in \eqref{eq:IMvalidity}, the following is true for all $\alpha \in [0,1]$: 
    \begin{enumerate}
        \item The degrees of belief and plausibility for all claims $C 
        \subseteq \mathbb{T}$ are calibrated, i.e.,
        \begin{equation}\label{eq:valNecPos}
            \sup_{\theta \notin C}\prob_\theta^n \left\{ \lPi_{Z^n}(C) \geq 1-\alpha \right \} \leq \alpha \quad \text{and} \quad \sup_{\theta \in C}\prob_\theta^n \left\{ \uPi_{Z^n}(C) \leq \alpha \right \} \leq \alpha.
        \end{equation}
        \item The $\alpha$ level sets of the IM's contour, 
$C_\alpha(z^n) = \{\theta \in \mathbb{T}: \pi_{z^{n}}(\theta) > \alpha\}$,
are nominal set estimates for $\Theta$ in the sense that $\prob^n_\Theta\bigl\{ C_\alpha(Z^{n}) \not\ni \Theta \bigr\} \leq \alpha.$
        \item The degrees of belief and plausibility are uniformly calibrated. For example, the degrees of plausibility satisfy
\begin{equation}\label{eq:uniform.validity}
\sup_{\theta \in C}\prob^n_\theta\bigl\{ \uPi_{Z^{n}}(C) \leq \alpha \text{ for some $C$ with $C \ni \theta$} \bigr\} \leq \alpha.    
\end{equation}        
\end{enumerate}
\end{cor}
As expected, the first inequality of \eqref{eq:valNecPos} confirms that the IMs degrees of belief satisfy the calibration condition in \eqref{eq:NeceVal}. The second inequality of \eqref{eq:valNecPos}, in turn, establishes that the IMs degrees of plausibility are also calibrated, in the sense that assigning low plausibility to a true claim about $\Theta$ is a rare event under $Z^n \sim \prob^n_{\Theta}$. This implies that the IMs degrees of plausibility can be used for hypothesis testing with guaranteed error rate control. Relatedly, property 2 shows that the IMs approach can also output calibrated set estimates. Property 3 strengthens IM calibration. The event ``for some $C$ with $C \ni \theta$'' corresponds to a union over all such sets, making it broader than any fixed claim. Hence the bound in \eqref{eq:uniform.validity} is stronger than that in \eqref{eq:valNecPos}, ensuring error control even when the analyst selects claims after observing the data. See \citet{CellaMartinMainSeverity} for further discussion.

There are situations in which interest focuses on a feature $\Phi = g(\Theta)$, where $g$ is a function defined on $\TT$. A marginal IM for $\Phi$ can be obtained from that for $\Theta$ through
\[\pi_{z^n}^g(\phi) = \sup_{\theta: g(\theta) = \phi}\pi_{z^n}(\theta), \quad \phi \in g(\TT).\]
Another consequence of the validity property in \eqref{eq:IMvalidity} is that this marginal IM inherits the desirable properties of the original IM construction for $\Theta$ stated in Corollary~\ref{cor:Main}.

We describe now what has been the standard construction of a data-dependent possibility contour that satisfies \eqref{eq:IMvalidity} in the present context where a parametric statistical model is assumed for the data. Let $L_{z^n}(\theta)$ denote the likelihood function based on the data $z^n$, and let $\hat{\theta}_{z^n}$ be the corresponding maximum likelihood estimate (MLE). Let $\rho(z^n,\theta)$ be the likelihood ratio evaluated at some particular value $\theta$ of $\Theta$, i.e., $\rho(z^n,\theta)= L_{z^n}(\theta) / L_{z^n}(\hat{\theta}_{z^n})$, which acts as the ranking function that measures the compatibility between $z^n$ and $\theta$. The relative likelihood is itself a data-dependent possibility contour, and inferences based on it have been extensively studied \citep[e.g.,][]{shafer1982, wasserman1990b, denoeux2006, denoeux2014}. There is no guarantee, however, that $\rho(Z^n,\Theta)$ satisfies the validity property in \eqref{eq:IMvalidity}. For this reason, the following validification is performed:
   \begin{equation}\label{eq:validification}
        \pi_{z^n}(\theta) = \prob^n_{\theta} \{ \rho(Z^n,\theta) \leq \rho(z^n,\theta)\}. 
    \end{equation}
This step transforms $\rho(z^n,\theta)$ into a possibility contour that is guaranteed to satisfy \eqref{eq:IMvalidity}. 

In summary, the {\em likelihood-based} IMs construction above requires the ability to compute the likelihood ratio and to evaluate its sampling distribution. While this distribution is rarely available in closed form, the assumption that $Z^n \iid \prob_\Theta$ allows for its approximation via Monte Carlo at each possible value of $\Theta$ on a sufficiently fine grid \citep[e.g.,][]{ryanpp2, hose.hanss.martin.belief2022}.  That is, the possibility contour in \eqref{eq:validification} is approximated as 
\begin{equation}\label{eq:StandIMsMC}
\pi_{z^n}(\theta) \approx \frac1L \sum_{\ell=1}^L \mathbbm{1} \{ \rho(Z_{\ell,\theta}^n,\theta) \leq \rho(z^n,\theta) \}, \quad \theta \in \TT,   
\end{equation}
where $\mathbbm{1}(\cdot)$ is the indicator function and $Z_{\ell,\theta}^n$ consists of $n$ iid samples from $\prob_\theta$, $\ell=1,\ldots,L$. 

The likelihood-based IM construction is powerful, as demonstrated by the properties summarized in Corollary~\ref{cor:Main}. However, its reliance on the likelihood function and the computation of MLEs can make it infeasible when the likelihood is unavailable or computationally intractable, or the MLE is difficult to obtain. This problem is compounded because validification typically uses the Monte Carlo approximation in \eqref{eq:StandIMsMC} and thus requires many likelihood-ratio evaluations. One further limitation of the Monte Carlo approximation is that the finite-sample validity guarantees of the IM become approximate whenever \eqref{eq:StandIMsMC} is used due to Monte Carlo error, and only hold exactly as $L \to \infty$.

\subsection{Key insights towards likelihood-free IMs} \label{ss:lfIMs:motivation}

Luckily, it turns out that the properties in Corollary~\ref{cor:Main} are not tied to the use of the likelihood ratio in the IMs construction. The general framework proposed in \citet{ryanpp1, imreview} shows that IMs are built through two key steps: the {\em ranking step}, where a real-valued function $\rho: \mathbb{Z}^n \times \TT \to \mathbb{R}$ that measures the compatibility between candidate values of $\Theta$ and the observed data $z^n$ (larger values indicating greater compatibility) is chosen; and the {\em validification step}, where this ranking function is rescaled via \eqref{eq:validification} to produce a possibility contour that satisfies \eqref{eq:IMvalidity}.

The crucial point is that the ranking function $\rho$ need not be the likelihood ratio; any real-valued measure of data--parameter agreement can serve in principle. The choice of the likelihood ratio in the IM ranking step has become standard because it tends to order parameter values sensibly across many problems. Moreover, it has recently been shown to be asymptotically efficient \citep{martin2025asymptotic}. That flexibility, however, comes with two demands: first, the chosen function must meaningfully capture data--parameter agreement; second, it must be computationally tractable. The efficiency of the resulting inferences hinges on the first of these, which underscores the appeal of the likelihood ratio as a default choice. Yet meaningfulness alone is not enough: the function must also be feasible to compute. After all, this is precisely why we consider alternatives to the likelihood ratio. 

It is hard to imagine a principled and computationally tractable ranking function $\rho$ that is general enough to replace the likelihood ratio in every—or even most—cases. More realistically, the choice of such a function will depend on the context of the problem at hand. For example, \citet{HectorCellaMartin2025} use a Gaussian ``working likelihood'' in the ranking step for non-Gaussian data to construct an IM within a divide-and-conquer framework. In what follows, we propose a general strategy for a likelihood-free IM construction. This is made possible by employing a fundamentally different yet intuitively sound approach to ranking candidate parameter values based on the observed data—one that does not take the form of $\rho$ but is instead based on data simulated from the model. We show that computational feasibility follows naturally from this alternative ranking. Importantly, this new construction still follows the core IM blueprint of ranking followed by validification, and our only assumption in the remainder of this paper is that it is feasible to generate samples from $\prob_\theta$.

\section{Likelihood-free IMs}
\label{s:lfIMs}

\subsection{Construction} \label{ss:lfIMs:construction}

The overall idea of this new likelihood-free construction is as follows: start with a low-dimensional summary statistic of $Z^n$, which is not necessarily an estimator of the parameter $\Theta$. Let $s_{\text{obs}}$ denote its observed value, based on data $z^n$. Because it is easy to simulate from the assumed model, we can generate, for any given candidate value $\theta$ of $\Theta$, a set of statistics 
$s^M(\theta) = \{ s_1(\theta), \ldots, s_M(\theta)\}$ by simulating $M$ data sets from $\prob_\theta^n$. The task is then to identify values of $\theta$ that produce simulated statistics resembling $s_{\text{obs}}$. To this end, we define a ranking function that, analogous to $\rho$ in the standard IM construction in Section \ref{s:IMs}, assigns higher values when the elements of $s^M(\theta)$ are close to $s_{\text{obs}}$ and lower values otherwise. This constitutes the ranking step. The transformation of this ranking function into a valid possibility contour through validification faces a new challenge, however, since Monte Carlo samples are now used in the ranking step—unlike in the standard IMs formulation, where such samples may be used to assist validification itself when the sampling distribution of $\rho$ is unknown; see \eqref{eq:StandIMsMC}. We show that validification in the new construction has a closed form if the ranking function is chosen to satisfy a key property, namely {\em permutation invariance}. We formalize these points next.

\subsubsection{Ranking step}
Let $d$ be the dimension of $s_{\text{obs}}$. Consider $\Psi_\theta: \mathbb{R}^{d M} \times \mathbb{R}^d \to \mathbb{R}$, a mapping that takes as input the collection of $M$ generated statistics $s^M(\theta)$ and the observed statistic $s_{\text{obs}}$. For notational simplicity, we henceforth write
\begin{align*}
\Psi_\theta \{ s^M(\theta),s_{\text{obs}} \} = \Psi_\theta ( s^M,s_{\text{obs}} ).
\end{align*}
We assume that $\Psi_\theta$ behaves like a \emph{conformity measure} in conformal prediction \citep{Vovk:2005}, in the sense that larger values indicate greater compatibility between its two arguments. In this way,
the essence of the IM ranking step is captured, since $\Psi_\theta (s^M,s_{\text{obs}} )$ measures the compatibility between the observed data and a candidate value $\theta$ of the unknown parameter $\Theta$. But unlike $\rho(z^n,\theta)$ in the standard IM construction, which directly compares them, $\Psi_\theta (s^M,s_{\text{obs}} )$ does so indirectly, by comparing the observed data to synthetic data generated under the assumption $\Theta = \theta$.

With an eye on validification, the only additional constraint we impose is that $\Psi_\theta$ is permutation-invariant in its first argument; that is, it is symmetric in the elements of $s^M(\theta)$, so that they may be arbitrarily permuted without affecting the value of $\Psi_\theta$. A simple example of such a permutation-invariant function suitable for univariate problems is 
\[\Psi_\theta (s^M,s_{\text{obs}} ) = - \left|s_{\text{obs}} - \bar{s}_\theta\right|,\]
where $\bar{s}_\theta$ is the average of $s^M(\theta)$. 
More generally, most statistical depth functions are permutation-invariant and assign larger values to more central summaries. In the examples of Section~\ref{s:examples}, we particularly use the Mahalanobis depth
\[\Psi_\theta\left(s^M,s_{\text{obs}}\right)= \left[1 + \left(s_{\text{obs}} - \bar{s}_\theta\right)^\top \widehat{\Sigma}_\theta^{-1} \left(s_{\text{obs}} - \bar{s}_\theta\right) \right]^{-1}, \]
where $\widehat{\Sigma}_\theta$ is the sample covariance matrix of $s^M(\theta)$, and the Tukey (halfspace) depth, which measures how central $s_{\text{obs}}$ is within the cloud of simulated summaries $s^M(\theta)$, defined as the smallest fraction of simulated summaries contained in any halfspace that includes $s_{\text{obs}}$. In practice, it can be computed using standard exact or approximate halfspace-depth algorithms; see, e.g., \citet{DyckerhoffMozharovskyi2016} and the references therein.

In summary, the proposed ranking function is both principled and computationally attractive. It assigns greater plausibility to parameter values that reproduce the observed summaries, while relying solely on simulation and simple depth evaluations. No likelihood computation or optimization is required.

\subsubsection{Validification step}

Let $s_{\text{obs}} = s_{M+1}$ and write $s^{M+1}(\theta) = \{ s^M(\theta), s_{M+1} \}$. For a fixed $\theta$, let $S_1(\theta), \ldots, S_{M+1}(\theta)$ denote random summary statistics computed from $M+1$ independent datasets generated under $\prob^n_\theta$, and $S^M(\theta) = \{ S_1(\theta), \ldots, S_M(\theta)\}$. For notational simplicity, we write $\Psi_\theta \{ S^M(\theta), S_{M+1} (\theta) \} = \Psi_\theta ( S^M, S_{M+1} )$. Having specified the ranking function $\Psi_\theta$, its validified counterpart would naturally be obtained by
\begin{equation}\label{eq:LFValidification}
\delta_{s^{M+1}}(\theta) := \delta_{s^{M+1}(\theta)}(\theta) = \prob_\theta^{n \times (M+1)}\bigl\{ \Psi_\theta ( S^M , S_{M+1}) \leq \Psi_\theta(s^M,s_{M+1}) \bigr\}.
\end{equation}
To avoid confusion, we emphasize that $S_{M+1}(\theta)$ is a random variable generated under $\prob_\theta^n$, whereas $s_{M+1}$ is the fixed observed summary and does not depend on $\theta$. The direct evaluation of the above probability is generally infeasible, since the distribution of $\Psi_\theta (S^M ,S_{M+1} )$ is almost surely unknown. This, however, is not unfamiliar territory within the IM framework. As discussed in Section~\ref{s:IMs}, the validification of the chosen ranking function can be approximated via Monte Carlo as in \eqref{eq:StandIMsMC}. 

The present setting is somewhat particular because $\Psi_\theta$ itself is already based on Monte Carlo samples. In principle, nothing prevents introducing an additional layer of $L$ Monte Carlo samples for its validification. There is a computational difficulty, however: for each $\theta \in \TT$, this strategy would require $M \times L$ Monte Carlo evaluations, and the overall cost escalates rapidly even for low-dimensional $\Theta$. One main motivation behind a likelihood-free IM is precisely to alleviate computational burdens arising from an intractable likelihood function, and our choice of ranking function achieves this goal. Replacing that burden with another would defeat the purpose. Further, validity is only guaranteed for $L \to \infty$. Therefore, a new strategy is needed to carry out the validification in \eqref{eq:LFValidification}.

IM validification in problems involving nuisance parameters has long been recognized as computationally challenging; see \citet{Ryanpp3,imreview}. A technique proposed to address this issue is \emph{conditioning}, in which the probability in the validification step is conditioned on a suitable function of the observed data that effectively eliminates nuisance parameters from the validification calculation, thereby significantly reducing computational cost. Although the context is different here—the full $\Theta$ is of interest—this strategy can still be adapted to render the probability calculation in \eqref{eq:LFValidification} tractable, thanks to the permutation-invariant nature of $\Psi_\theta$.

The joint distribution of $\{S_1 (\theta), \ldots, S_{M+1} (\theta) \}$ is exchangeable since each summary statistic is computed from an independent dataset generated under $\prob^n_\theta$. This structure allows us to condition on the unordered collection $\{s_1, \ldots, s_M, s_{M+1}\}$. More specifically, the proposed validified ranking—based on conditioning—is given by
\begin{align*}
\delta_{s^{M+1}}(\theta) &= \prob_\theta^{n \times (M+1)}\left[\Psi_\theta(S^M,S_{M+1}) \leq \Psi_\theta(s^M,s_{M+1}) \, | \, \{s_1, \ldots, s_M, s_{M+1}\} \right] \\
&= \frac{1}{(M+1)!}\sum_\sigma \mathbbm{1}\{\Psi_\theta(s^{\sigma(1:M)},s_{\sigma(M+1)}) \leq \Psi_\theta(s^M,s_{M+1})\}, 
\end{align*}
where the sum is over all $(M + 1)!$ permutations, $\sigma$, of the integers $1, \ldots, M, M + 1$. Because $\Psi_\theta$ is permutation-invariant in its first argument, the value of $\Psi_\theta(s^{\sigma(1:M)},s_{\sigma(M+1)})$ depends only on which element is singled out as the observed statistic. Consequently, the right-hand side above can be further simplified:
\begin{equation}\label{eq:validified}
\delta_{s^{M+1}}(\theta)= \frac{1}{M+1}\sum_{i=1}^{M+1}\mathbbm{1}\{\Psi_\theta(s^{M+1}_{-i},s_i) \leq \Psi_\theta(s^M,s_{M+1})\},
\end{equation}
where $s^{M+1}_{-i}$ denotes $s^{M+1}(\theta)$ with the $i$th element removed. Validification in the proposed likelihood-free IM construction therefore reduces to averaging over the $M+1$ possible choices of which statistic plays the role of the observed one. This simple, closed-form counting solution achieves the computational simplicity sought, keeping the proposed IM efficient not only in the ranking step but throughout.

The reader may have noticed that we have not referred to the validified ranking in \eqref{eq:validified} as a contour. This is because there is no guarantee that $\sup_{\theta \in \TT} \delta_{s^{M+1}}(\theta)=1$ for all $s_{\text{obs}}$, implying that an additional normalization step may be necessary in the proposed construction. The likelihood-based construction in Section~\ref{s:IMs} does not require such a step because it is guaranteed that $\pi_{z^n}(\hat\theta_{z^n})=1$, where $\hat\theta_{z^n}$ denotes the MLE. For the depth-type rankings considered in our likelihood-free IM construction, \eqref{eq:validified} reaches the value one if, for some $\theta \in \TT$, the observed summary $s_{\text{obs}}$ is the most ``central'' element among the $M+1$ summaries in the leave-one-out comparisons. For example, under the Mahalanobis depth this occurs when $s_{\text{obs}}$ lies closest to the center of the simulated summaries relative to their empirical covariance structure. Under the Tukey depth, it occurs when $s_{\text{obs}}$ lies deepest within the cloud of simulated summaries, in the sense that any halfspace containing $s_{\text{obs}}$ also contains a large fraction of the simulated summaries.
Although these are very plausible for some $\theta \in \TT$, it is not guaranteed. Therefore, the validified likelihood-free IM contour is formally defined as
\begin{equation}\label{eq:valnorm}
 \pi_{s^{M+1}}(\theta) = \frac{\delta_{s^{M+1}}(\theta)}{\sup_{\vartheta \in \TT} \delta_{s^{M+1}}(\vartheta)}, \quad \theta \in \TT.    
\end{equation}

This additional normalization step is not unfamiliar in the IM framework; see, for example, Section~7 of \citet{CELLAMARTINREGPREDIJAR} and \citet{CellaFusion}. However, unlike those cases, we emphasize that in the present context normalization rarely needs to be applied, as it is typically the case that $\sup_{\theta \in \TT} \delta_{s^{M+1}}(\theta)=1$, so that $\pi_{s^{M+1}}(\theta)=\delta_{s^{M+1}}(\theta)$. In fact, in all examples and simulations considered in Sections~\ref{s:examples} and~\ref{s:real}, normalization was never required.

\subsection{Finite-sample guarantees}
The contour in \eqref{eq:valnorm} constitutes the basic output of the proposed likelihood-free IM. Degrees of belief and plausibility for any claim of interest $C \subseteq \TT$ are obtained from this contour via the usual possibility calculus:
\[\uPi_{s^{M+1}}(C) = \sup_{\theta \in C} \pi_{s^{M+1}}(\theta) \quad \text{and} \quad \lPi_{s^{M+1}}(C) = 1 - \uPi_{s^{M+1}}(C^c).\]
These probabilistic assignments are only meaningful, however, if \eqref{eq:valnorm} is properly calibrated. Recall from Section~\ref{s:IMs} that IM validity requires the contour, viewed as a function of the problem’s random elements, to be stochastically no smaller than a \unif(0,1) random variable. Below, we establish this property for the proposed likelihood-free IM, beginning with an auxiliary result. To this end, define $T_i(\theta) = \Psi_\theta(S^{M+1}_{-i},S_i)$, $i = 1, \ldots, M+1$ and $\operatorname{rank}\{T_{M+1}(\theta)\} = \sum_{i=1}^{M+1} \mathbbm{1}\{T_i(\theta) \leq T_{M+1}(\theta)\}$.

\begin{lem}\label{lemma:1} 
For all $k=1,\ldots,M+1$,
\[\prob_\Theta^{n \times (M+1)}\{\operatorname{rank}\{T_{M+1}(\Theta)\} \le k\} \le k/(M+1),\]
i.e., $\operatorname{rank}\{T_{M+1}(\Theta)\}$ is stochastically no smaller than $\unif\{1, \ldots, M +1\}$.
\end{lem}
\begin{proof}
Under the true $\Theta$, the components of $S^{M+1}(\Theta)$ are iid. Assume there are no ties among $T_1(\Theta), \ldots, T_{M+1}(\Theta)$. Due to the permutation-invariant nature of $\Psi_\Theta$, their ordering is uniformly random, so it follows that the above inequality is an equality. If ties are possible, $\operatorname{rank}\{T_{M+1}(\Theta)\}$ is stochastically no smaller than when ties are not possible,
which establishes the result. 
\end{proof}

\begin{thm}\label{thm:validity}
The likelihood-free IM for $\Theta$ established above, with contour as in \eqref{eq:valnorm}, is valid in the sense that
\begin{equation}\label{eq:MCIMvalidity}
\prob_\Theta^{n \times (M+1)}\{ \pi_{S^{M+1}}(\Theta) \leq \alpha\} \leq \alpha, \quad \text{for all $\alpha \in [0,1]$}.
\end{equation}
\end{thm}
\begin{proof}
Start by noticing that the validified ranking in \eqref{eq:validified} can be rewritten as
\begin{align*}
\delta_{s^{M+1}}(\theta) &= \frac{1}{M+1} \sum_{i=1}^{M+1} \mathbbm{1}\{t_i(\theta) \leq t_{M+1}(\theta)\} = \frac{1}{M+1} rank\{ t_{M+1}(\theta)\}.
\end{align*}
By Lemma~\ref{lemma:1}, $\operatorname{rank}\{T_{M+1}(\Theta)\}$ is stochastically no smaller than $\unif\{1, \ldots, M+1\}$, so
\[\delta_{S^{M+1}}(\Theta)=\frac{\operatorname{rank}\{T_{M+1}(\Theta)\}}{M+1}\]
is stochastically no smaller than $\unif\left\{\frac{1}{M+1}, \frac{2}{M+1}, \ldots,1\right\}$, which is itself stochastically no smaller than $\unif(0,1)$. Since \eqref{eq:valnorm} implies $\pi_{s^{M+1}}(\theta) \geq \delta_{s^{M+1}}(\theta)$ for all $\theta \in \TT$, it follows that
\[\prob_\Theta^{n \times (M+1)}\{ \pi_{S^{M+1}}(\Theta) \leq \alpha\} \leq  \prob_\Theta^{n \times (M+1)}\{ \delta_{S^{M+1}}(\Theta) \leq \alpha\} \leq  \alpha,\]
which establishes the result.
\end{proof}

What distinguishes the finite-sample validity of the likelihood-free IM contour stated in \eqref{eq:MCIMvalidity} from that of the likelihood-based IM contour in \eqref{eq:IMvalidity} is that our proposal also accounts for uncertainty due to the Monte Carlo samples used in the ranking step. Pure finite-sample validity of the likelihood-based construction holds only when the sampling distribution of the likelihood ratio is known.
Since this is rarely the case, the Monte Carlo approximation in \eqref{eq:StandIMsMC} must usually be employed, making finite-sample validity only approximate—although, in practice, this is typically not a concern when $L$ is large. On the other hand, the likelihood-free approach proposed here achieves finite-sample validity regardless of the size of $M$. The motivation for choosing a large $M$ in this setting is to improve the efficiency of the resulting inferences: comparing the observed sample against a larger number of Monte Carlo replicates provides a clearer picture of which values of $\Theta$ are plausible than when the comparison is based on only a few replicates.
From a more theoretical perspective, note that as $M$ increases, the distribution of $\pi_{S^{M+1}}(\Theta)$ approaches \unif(0,1), which represents the upper bound of efficiency for an IM contour.

While efficiency can be improved by increasing $M$, it also depends on another factor—the choice of summary statistic. A key advantage of the proposed IM is that its finite-sample validity holds for any summary, regardless of its dimension. This is a powerful property, since a finite-sample valid IM can be constructed even when only a summary of the data is available, rather than the full sample. When the full data are available and the analyst can choose the summary, efficiency considerations become important. In general, more efficient inference can be expected when (i) the summary has a dimension matching that of the parameter of interest, and (ii) the summary captures more information about the parameter—for instance, when it is a sufficient statistic. These points are illustrated in Section~\ref{ss:cor}. In summary, efficiency can be pursued along two complementary directions: by selecting an appropriately informative summary statistic and, given this statistic, by increasing the number of Monte Carlo samples $M$.

The following corollary establishes the key features of our likelihood-free IM that result from the validity property in \eqref{eq:MCIMvalidity}, analogous to the ones established  in Corollary~\ref{cor:Main} for the likelihood-based IM.

\begin{cor}\label{cor:Main_LF}
    For an IM based on the contour in \eqref{eq:valnorm}, the following is true for all $\alpha \in [0,1]$:  
    \begin{enumerate}
        \item The degrees of belief and plausibility for all claims $C 
        \subseteq \mathbb{T}$ are calibrated, i.e.,
        \begin{equation}\label{eq:lower_upperLFIM}
          \sup_{\theta \notin C}\prob_\theta^{n \times (M+1)} \left\{ \lPi_{S^{M+1}}(C) \geq 1-\alpha \right \} \leq \alpha \quad \text{and} \quad \sup_{\theta \in C}\prob_\theta^{n \times (M+1)} \left\{ \uPi_{S^{M+1}}(C) \leq \alpha \right \} \leq \alpha.
          \end{equation}
        \item The $\alpha$ level sets of the IM's contour, 
$C_\alpha\{{s^{M+1}(\theta)}\} = \{\theta \in \mathbb{T}: \pi_{s^{M+1}}(\theta) > \alpha\}$,
are nominal set estimates for $\Theta$ in the sense that
\[ \prob^{n\times(M+1)}_\Theta\bigl[ C_\alpha\{{S^{M+1}(\Theta)}\} \not\ni \Theta \bigr] \leq \alpha. \]
        \item The degrees of belief and plausibility are uniformly calibrated. For example, the degrees of plausibility satisfy
\[\sup_{\theta \in C}\prob^{n\times(M+1)}_\theta\bigl\{ \uPi_{S^{M+1}}(C) \leq \alpha \text{ for some $C$ with $C \ni \theta$} \bigr\} \leq \alpha.\]           
\end{enumerate}
\end{cor}
\begin{proof}
The following hold for the corresponding statements above:
\begin{enumerate}
    \item Consider the right-hand side expression (an analogous argument can be made for the left-hand side one). For any claim $C$ that contains $\Theta$ we have that $\uPi_{s^{M+1}}(C) \geq \pi_{s^{M+1}}(\Theta)$ due to the monotonicity of the possibility contour. 
\item $C_\alpha\{ s^{M+1}(\Theta) \} \not\ni \Theta$ if and only if $\pi_{s^{M+1}}(\Theta) \leq \alpha$.
\item There exists a $C$ such that $C \ni \Theta$ and $\uPi_{s^{M+1}}(C) = \sup_{\Theta \in C} \pi_{s^{M+1}}(\Theta) \leq \alpha$ if and only if $\pi_{s^{M+1}}(\Theta) \leq \alpha$.
\end{enumerate}
These facts, combined with \eqref{eq:MCIMvalidity} in Theorem \ref{thm:validity}, give the results.
\end{proof}

As discussed in Section~\ref{s:IMs}, in the standard IM construction, the validity of the contour for $\Theta$, together with the possibility calculus based on optimization, yields marginal IMs for any feature $\Phi = g(\Theta)$ that retain all the calibration guarantees stated in Corollary~\ref{cor:Main}. It is therefore natural to expect that this reliable marginalization can also be achieved within the likelihood-free construction proposed here. Corollary~\ref{cor:LF_marg} confirms this. 

\begin{cor}\label{cor:LF_marg}
    Let $\Phi = g(\Theta)$ be a feature of interest, where $g$ is a function defined on $\TT$, and let $\pi_{s^{M+1}}$ be a valid likelihood-free IM contour for $\Theta$ in the sense of \eqref{eq:MCIMvalidity}. Then the marginal IM for $\Phi$ with contour
    \[\pi_{s^{M+1}}^g(\phi) = \sup_{\theta: g(\theta) = \phi}\pi_{s^{M+1}}(\theta), \quad \phi \in g(\TT),\]
is valid in the sense that  
\[\prob_\Phi^{n \times (M+1)}\{ \pi_{S^{M+1}}^g(\Phi) \leq \alpha\} \leq \alpha, \quad \text{for all $\alpha \in [0,1]$}.\]
\end{cor}
\begin{proof}
Since $\pi_{s^{M+1}}^g$ is computed by optimizing $\pi_{s^{M+1}}$ and $\pi_{S^{M+1}}(\Theta)$ satisfies \eqref{eq:MCIMvalidity}, it then follows that $\pi_{S^{M+1}}^g(\Phi)$ is also stochastically no smaller than \unif(0,1), establishing the result.
\end{proof}

The properties in Corollary~\ref{cor:Main_LF} carry over to marginal IMs, ensuring that users obtain (uniformly) calibrated lower and upper probabilities for claims of interest about $\Phi$, together with nominal set estimates.
The proof is omitted, since the arguments are analogous to those used in Corollary~\ref{cor:Main_LF}.
Marginal calibration guarantees are crucial, as they protect data analysts from unintentional errors in subsequent inferences; see Section~\ref{ss:gk}.

\section{Synthetic data examples}
\label{s:examples}

\subsection{Correlation coefficient}
\label{ss:cor}

Suppose we're working with bivariate normal data with zero means, unit variances, and an unknown correlation coefficient $\Theta$. Figure~\ref{fig:cor}(a) shows the likelihood-free IM contour (in blue) for one dataset of size $n = 30$, constructed using the sample correlation as the summary statistic, Mahalanobis depth, and $M=1{,}000$. For comparison, the black contour shows what we would obtain if we had access to the full dataset and used the likelihood-based IM construction. While this is not an extreme case in which the MLE is extremely difficult to obtain, its evaluation still requires numerical computation (recall that the MLE is not the sample correlation coefficient). To confirm the validity of the proposed likelihood-free solution, $1{,}000$ datasets of size $n = 30$ were generated from $\Theta = 0.5$. For each repetition, the possibility contour at $\Theta = 0.5$ was evaluated. Figure~\ref{fig:cor}(b) shows the empirical distribution function of these possibilities, confirming its uniform distribution.

Figure~\ref{fig:cor}(a) shows that there is no significant loss of efficiency when comparing the likelihood-free IM based on the correlation coefficient and the likelihood-based IM. We argued in Section~\ref{ss:lfIMs:construction} that, towards efficiency, one should try to select a summary statistic with the same dimension as the parameter of interest, and that the more directly this statistic reflects the parameter, the better. To illustrate this, Figure \ref{fig:cor} also displays likelihood-free IM contours based on the two-dimensional sufficient statistic (red)—the sum of squared values from both variables and the sum of the cross–products of the paired observations—and on the single statistic given by the cross–product sum alone (green), both constructed using the Mahalanobis distance with $M=1{,}000$. Using the two-dimensional sufficient statistic, although it contains all relevant information about $\Theta$, results in inefficiency because of its unnecessarily high dimensionality. Conversely, relying on just one sufficient statistic is also inefficient because it discards part of the relevant information about $\Theta$. This explains why the sample correlation coefficient is particularly effective: it is a one-dimensional statistic that, in a sense, captures the essential information contained in the sufficient statistics. Figure~\ref{fig:cor}(b) also shows the empirical distribution function of the two less-efficient likelihood-free IMs in the same aforementioned simulation setting. It confirms that, aside from efficiency considerations, likelihood-free IMs are valid regardless of the choice of summary statistic.

\begin{figure}[t]
\begin{center}
\subfigure[Possibility contours.]{\scalebox{0.42}{\includegraphics{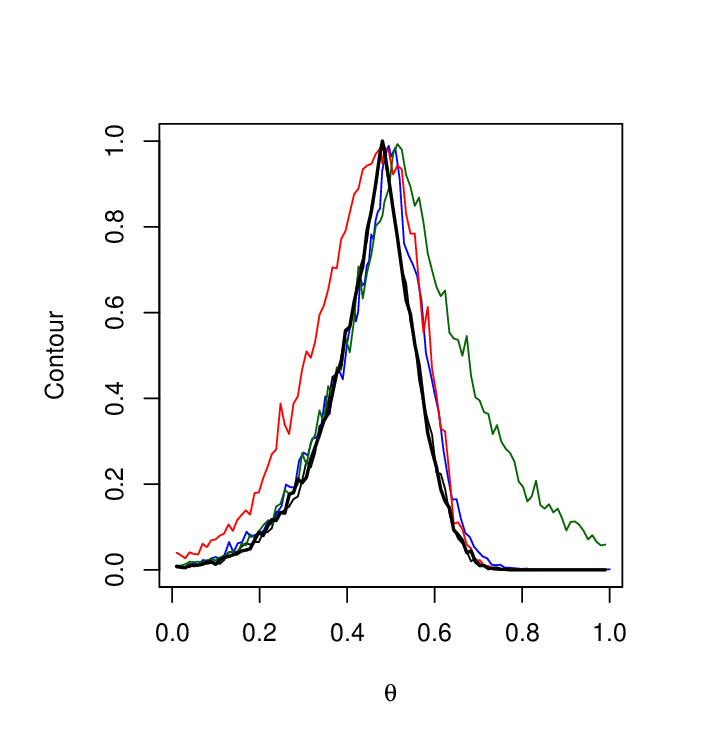}}}
\subfigure[Distribution functions of $\pi_{S^{M+1}}(\Theta)$.]{\scalebox{0.42}{\includegraphics{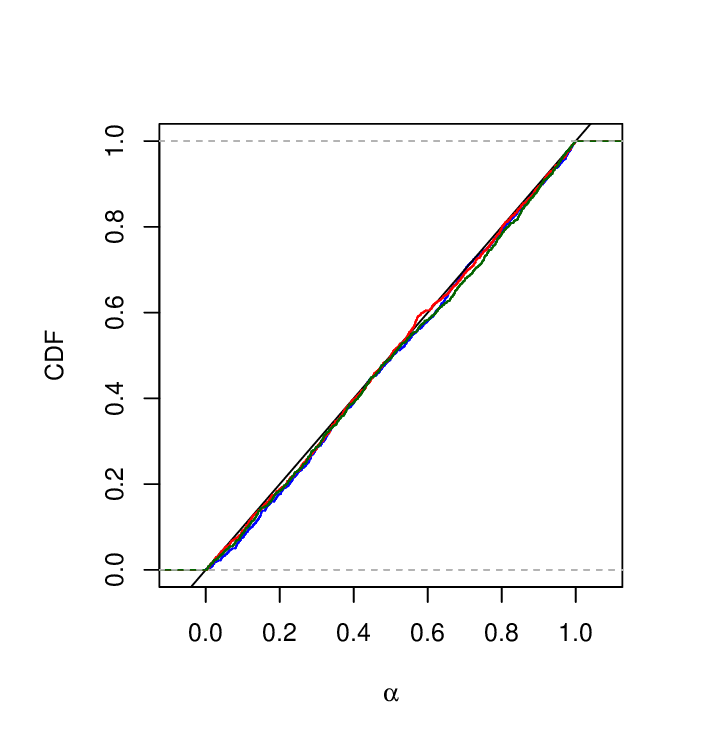}}}
\end{center}
\caption{Results shown in black correspond to the likelihood-based contour. Results shown in blue, red, and green correspond to the likelihood-free contours based on the sample correlation coefficient, the two-dimensional sufficient statistic, and a single sufficient statistic, respectively.}
\label{fig:cor}
\end{figure}

\subsection{$g$-and-$k$ distributions}
\label{ss:gk}

The $g$-and-$k$ family of distributions \citep{haynes1997robustness, rayner2002numerical} has been an appealing choice to model complex data in financial and climate applications, among others, due to its ability to capture asymmetry and heavy tails in a parsimonious way. The distribution is defined through its quantile function
\begin{align*}
Q(u)&= \mu + \sigma w_u \left( 1+c \, \frac{1-e^{-gw_u}}{1+e^{-gw_u}} \right) (1+w_u^2)^{k}, \quad 0 \leq u \leq 1,
\end{align*}
where $\mu \in \mathbb{R}$ is a location parameter, $\sigma > 0$ is a scale parameter, $g \in \mathbb{R}$ measures skewness and $k >-1/2$ measures kurtosis, $w_u = \Phi^{-1}(u)$ is the $u^\text{th}$ standard Gaussian quantile, and $c$ is a constant corresponding to the value of ``overall symmetry''. It is common to set $c=0.8$ \citep[see, e.g.][]{rayner2002numerical, drovandi2011likelihood}, and we adopt this convention throughout this section.

After observing data $z^n$, the goal is to construct an IM for $\Theta = \{\mu,\sigma,g,k\}$. The standard IMs construction based on the likelihood ratio faces major challenges since the density function of the $g$-and-$k$ distribution is not available in closed form (except in trivial cases). In fact, to get the likelihood one needs to solve the inverse problem $z_i = Q(u_i)$ for each observation $z_i$, $i= 1,\ldots, n$, which is both expensive and numerically challenging. Simulation from the model is, however, straightforward, requiring just a transformation of standard normal quantiles. This makes the proposed likelihood-free IM an appealing strategy for provably finite-sample valid probabilistic inferences on $\Theta$.

As an illustration, Figure~\ref{fig:gk}(a) shows the histogram of a dataset of size $n=100$ that is known to come from a $g$-and-$k$ distribution with $\mu=0$ and $\sigma=1$, but the values of the skewness and kurtosis parameters $g$ and $k$ are unknown. Towards a likelihood-free IM for $\Theta=(g,k)$, we use the sample skewness and sample kurtosis as a bivariate summary statistic, Tukey depth, and set $M=250$. Figure~\ref{fig:gk}(b) shows the 0.2-, 0.5-, and 0.8-level sets of the IM contour for $\Theta$. Because simulation from the $g$-and-$k$ model is straightforward, it has become a canonical example for demonstrating simulation-based inference methods, such as approximate Bayesian computation (ABC) \citep[e.g.,][]{fearnhead2012constructing}. For comparison, Figure~\ref{fig:gk}(b) also displays the 20\%, 50\%, and 80\% highest posterior density regions for $\Theta$, obtained via an ABC rejection sampler with weakly informative priors $g \sim \unif(-10,10)$ and $k \sim \unif(-0.5,10)$. The observed and simulated data were compared using the sample skewness and sample kurtosis, and the closest 1\% of parameter draws were retained. The joint posterior was then estimated using a kernel density approximation.
Note that while the Bayesian solution appears more efficient, this is meaningless if the Bayesian posterior is not properly calibrated.

To evaluate this, we conduct a simulation study following the setup described above, where $\Theta = (g, k)$. We generate 500 datasets, each consisting of $n = 100$ observations, using the true parameter values $\Theta_1 = 2$ and $\Theta_2 = 0.25$. We consider the following three claims:  
\[C_1: 3 < \Theta_1 < 4, \quad 
C_2: \Theta_2 < 0, \quad 
C_3: \Theta_1 / \Theta_2 > 10,\]
each concerning a different feature of $\Theta$.  
Note that {\em all three hypotheses are false}, i.e., the true parameter does not satisfy any of these constraints.  Figure~\ref{fig:gk}(c) shows the distribution functions of the lower probabilities assigned to these claims by the likelihood-free IM. As implied by the left-hand side of \eqref{eq:lower_upperLFIM}, calibration corresponds to the curves lying above the diagonal line, which is observed in all three cases.  
For comparison, Figure~\ref{fig:gk}(c) also displays the distribution functions of the posterior probabilities assigned to the same claims by an ABC rejection sampler with weakly informative priors, as described above. While the posterior probability for $C_2$ is properly calibrated, those for $C_1$ and $C_3$ are not: curves below the diagonal indicate that posterior probabilities for these false claims are stochastically larger than $\unif(0,1)$. In particular, $C_3$ is assigned posterior probabilities exceeding 0.5 about 70\% of the time. This lack of calibration is undesirable, as it can lead to ``systematically misleading conclusions'' \citep{ReidandCox2015}.

\begin{figure}[t]
\begin{center}
\subfigure[]{\scalebox{0.42}{\includegraphics{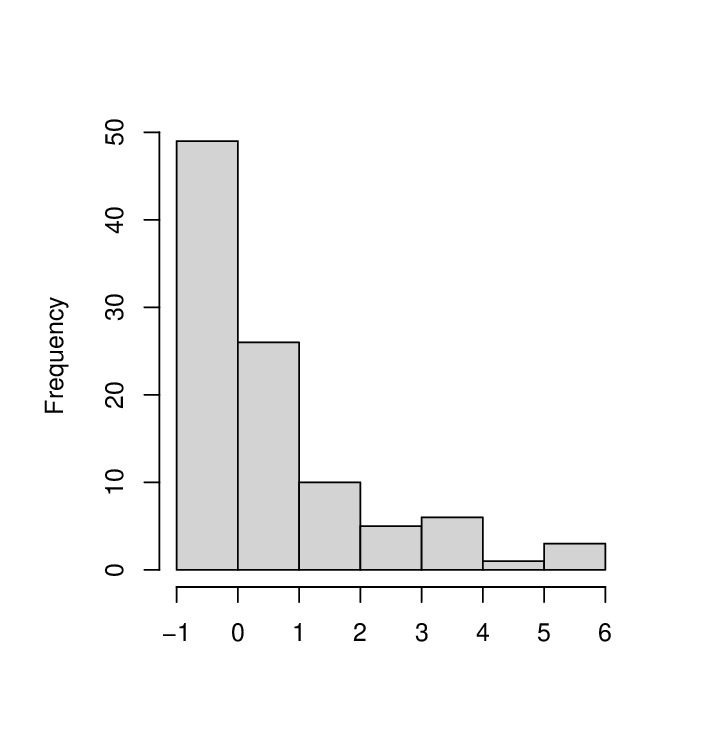}}}
\subfigure[]{\scalebox{0.42}{\includegraphics{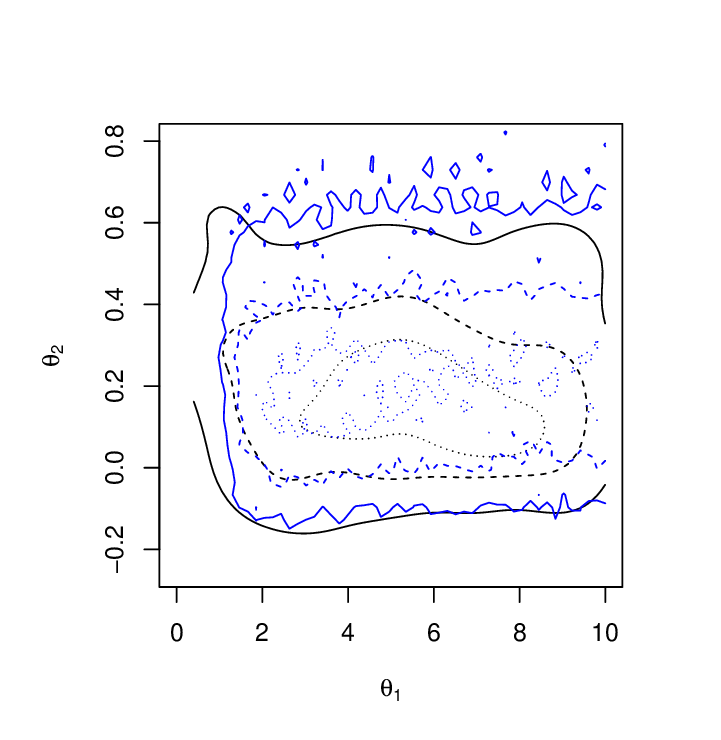}}}
\subfigure[]{\scalebox{0.42}{\includegraphics{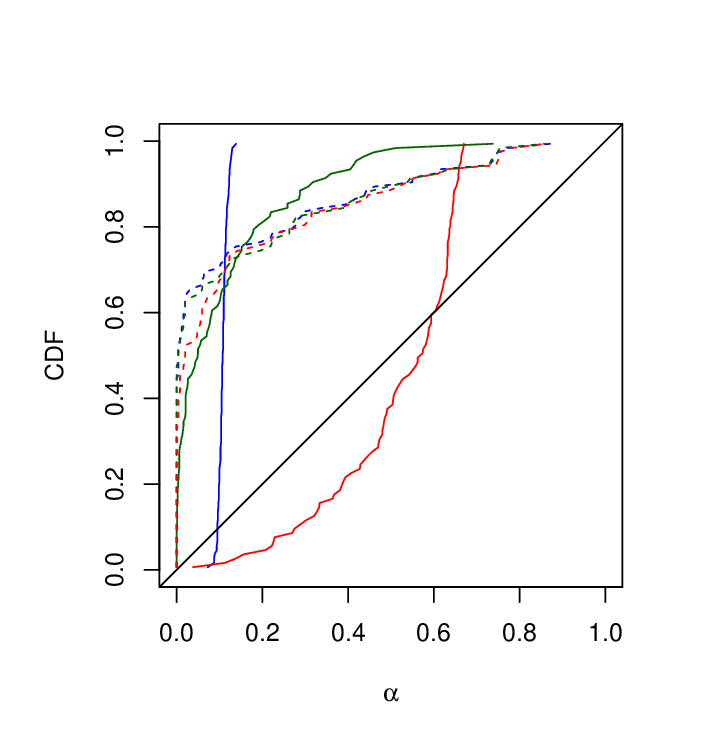}}}
\end{center}
\caption{Panel~(a): Histogram of a dataset, n=100. Panel~(b): In blue, the 0.2-, 0.5-, and 0.8-level sets of the likelihood-free IM contour. In black, the 20\%, 50\%, and 80\% highest posterior density regions of the ABC solution. Solid, dashed, and dotted lines correspond to increasing levels. Panel~(c): Distribution functions of
the IM lower probabilities (dashed) and Bayes posterior probabilities (solid) assigned to $C_1$ (blue), $C_2$ (green) and $C_3$ (red), based on 500 simulation replicates.}
\label{fig:gk}
\end{figure}

\subsection{Differential privacy}
\label{ss:dp}

The development of privacy-preserving methodologies has advanced substantially over recent decades, with differential privacy (DP) emerging as the leading framework for data protection. The formal definition of DP was introduced by \citet{Dwork2026}, and requires that the output distribution of a randomized algorithm remain nearly unchanged when a single individual’s record is added or removed from the dataset. More specifically, a mechanism is said to satisfy $\varepsilon$-differential privacy if, for any two datasets differing in one entry and for any measurable subset of outputs, the probability assigned to that subset changes by at most a multiplicative factor of $e^\varepsilon$. This definition provides strong worst-case privacy guarantees while remaining compatible with a wide range of statistical and machine learning tasks.

Analyzing privatized data is challenging because the addition of noise to ensure privacy alters the sampling distribution of estimators in non-standard ways that depend on the mechanism employed. As a result, classical inference procedures are generally invalid without adjustment \citep{karwa2018finite,rogers2021guide}. In the DP framework, however, the distribution of the added noise can be publicly disclosed without compromising privacy guarantees, allowing analysts to incorporate the additional uncertainty into their statistical reasoning.
Nevertheless, the marginal likelihood of parameters based on privatized statistics is often computationally intractable, making approximation methods—such as the parametric bootstrap \citep{ferrando2022parametric} or asymptotic arguments \citep{wang2018statistical}—the norm in applications, though these do not guarantee validity in finite samples. To address this limitation, \citet{Awan03012025} recently proposed a general-purpose frequentist simulation-based inference framework, built on repro samples \citep{xie2022repro}, that yields finite-sample valid confidence intervals and tests and can be applied to DP problems. On the probabilistic uncertainty quantification side, MCMC approaches targeting the posterior distribution given privatized data have also been explored \citep[e.g.,][]{ju2022dataaugmentationmcmc}, although the False Confidence Theorem raises concerns about the reliability of the resulting posterior summaries. This naturally raises the question: can an IM be constructed that is suitable for DP problems in general? If so, one could obtain the best of both worlds—frequentist-calibrated procedures together with meaningful probabilistic assessment—a contribution of significant value for the analysis of privatized data.

The aforementioned intractability of the likelihood function poses a challenge to the construction of likelihood-based IMs in a DP context. However, since the noise distribution is publicly available under DP, sampling from it is straightforward, making the likelihood-free IMs construction a natural general-purpose method for finite-sample valid probabilistic inference in DP settings. As an example, suppose $z^n$ is composed by iid Bernoulli random variables with unknown mean $\Theta$, and the goal is to construct an IM for $\Theta$. If the raw data or the sufficient statistic $\sum_{i=1}^n z_i$ were available, the derivation of a likelihood-based IM would be straightforward. Instead, assume the analyst observes only a privatized release $\sum_{i=1}^n z_i + N$, where $N$ is some independent noise random variable. 
Following \citet{Awan_Slavkovic_2020,Awan03012025}, we assume $N \sim \Tulap \{ 0,\exp(-1),0 \}$, where $\Tulap(m,b,q)$ denotes the Tulap (``truncated-uniform Laplace'') distribution. The parameter $m \in \mathbb{R}$ specifies location, while $b=\exp(-\epsilon)$ is determined by the privacy level $\epsilon$, and $q$ controls truncation. Choosing $N \sim \Tulap \{ 0,\exp(-1),0 \}$ guarantees 1-DP. Moreover, \citet{Awan_Slavkovic_2020} show that $N \stackrel{d}{=} G_1 - G_2 + U$,
where $G_1, G_2$ and $U$ are independent, with $G_1,G_2 \iid \geom \{ 1-\exp(-1) \}$ and $U\sim\unif(-1/2,1/2)$. This representation makes it straightforward to simulate the privatized statistic $\sum_{i=1}^n Z_i + N$ under different values of $\Theta$.

Figure~\ref{fig:dp} shows the likelihood-free IM contour for one data set of size $n=25$, obtained using the Mahalanobis depth and $M=1{,}000$. For comparison, the likelihood-based IM contour is also displayed. The likelihood can be derived here since the sampling distribution of the privatized statistic can be obtained through the convolution of the binomial distribution of $\sum_{i=1}^n z_i$ and the Tulap distribution of $N$; see \citet{Awan_Slavkovic_2020} for details. Moreover, the MLE can be obtained by numeric optimization. Note that the likelihood-free contour closely resembles the likelihood-based contour, and no efficiency appears to be lost by using the former, with the added benefit of not requiring an analytical sampling distribution or numerical optimization. To verify this, we conduct a simulation study in which the above scenario is repeated $1{,}000$ times. Table~\ref{tab:dp} reports the average widths of the 30\%, 60\%, and 90\% interval estimates obtained from both approaches, showing negligible differences. The table also presents the corresponding average coverages, confirming that the IM interval estimates are properly calibrated, whether based on the likelihood or constructed in a likelihood-free manner.

\begin{figure}[t]
  \centering
  \includegraphics[height=4.9cm,trim= 0 5 8cm 17cm,clip]{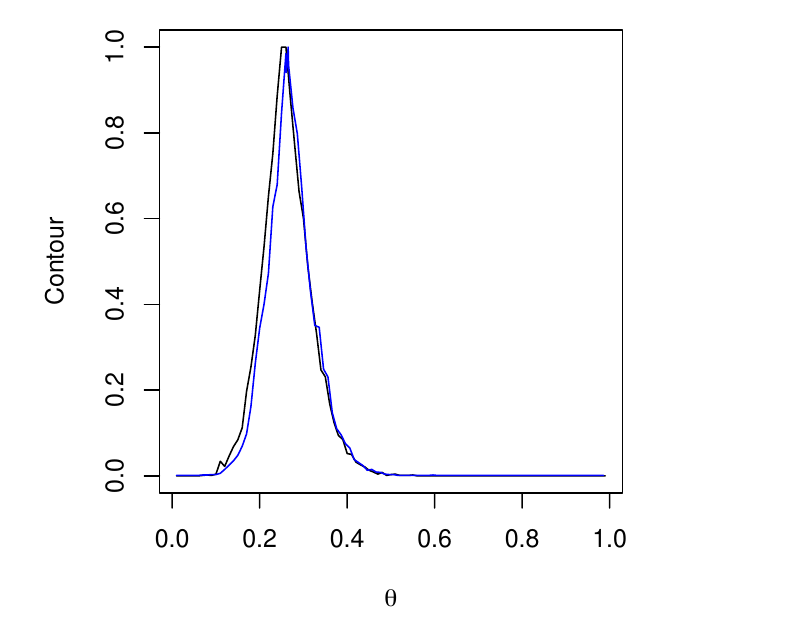} 
  \caption{Likelihood-based contour (black) and likelihood-free contour (blue).}
  \label{fig:dp}
\end{figure}

\begin{table}[t]
\caption{Average interval width ($\times 100$) with average coverage in parentheses} \label{tab:dp}
\centering
\begin{tabular}{lrrr}
\hline
\multirow{2}{*}{IM Contour} & \multicolumn{3}{c}{$100(1-\alpha)\%$} \\
 & $30$ & $60$ & $90$ \\
\hline
likelihood-based & 6.10 (0.32) & 12.74 (0.63) & 23.76 (0.90) \\
likelihood-free  & 5.35 (0.30) & 11.63 (0.58) & 22.45 (0.89) \\
\hline
\end{tabular}
\end{table}

\subsection{Ising model}

The Ising model, named after the physicist Ernst Ising \citep{ising1924}, has become a widely used framework for modeling binary data, finding applications in a wide range of areas, e.g., quantum computation \citep{lahtinen2017}, statistical genetics \citep{majewski2001}, spatial statistics \citep{besag1974}, among others. In its two-parameter form, the model is defined on an undirected graph in which each node is a binary variable $z_i \in \{-1,1\}$, $i=1,\dots,n$. The connections between $z_i$'s are determined by a symmetric coupling matrix $A_n \in \mathbb{R}^{n \times n}$. The likelihood of the Ising model is defined as the probability of the vector $z=(z_1, \ldots, z_n)^\top \in \{-1,1\}^n$:
\begin{equation}\label{eq:Ising_like}
\prob_{\beta,B}(Z=z) \propto \exp\!\Big( \tfrac{\beta}{2} z^\top A_n z + B \sum_{i=1}^n z_i \Big),    
\end{equation}
where $\beta > 0$ represents the interaction strength among $z_i$'s and $B\neq0$ represents external influence on $Z$. Following \citet{bhattacharya2018,ghosal2020,okabayashi2011,kim2024} we assume that $A_n$ is fully specified, so the parameter of interest is $\Theta=(\beta,B)$.

Any likelihood-based inference method faces challenges in making inferences for $\Theta$ due to the intractability of the normalizing constant in \eqref{eq:Ising_like}. More specifically, computing this constant requires summing over $2^n$ configurations, which is infeasible except for very small $n$. A common solution is to replace the true likelihood with a pseudo-likelihood; see \citet{bhattacharya2018,ghosal2020,okabayashi2011,kim2024}.
Simulation-based inference \citep[e.g.,][]{grelaud2009} is also an option: although direct likelihood evaluation is impractical, simulation from the model is possible through Markov chain Monte Carlo methods such as Gibbs sampling, since the normalizing constant cancels in the conditional probabilities and acceptance ratios \citep{besag1974}.

As a first illustration, consider the {\em zero-field} Ising model where the external influence parameter $B$ is set to zero, so that $\Theta=\beta$. We take $A_n$ to be the standard nearest-neighbor matrix, meaning that each site interacts only with its immediate horizontal and vertical neighbors; equivalently, $A_n$ has entries equal to one if two sites are nearest neighbors and zero otherwise. In such case, the right-hand side of \eqref{eq:Ising_like} simplifies to
$\exp (\beta \sum_{(i,j)\in E} z_i z_j )$,
where $E$ denotes the set of nearest--neighbor pairs. Even in this simplified setting, the normalizing constant remains intractable for moderate lattice sizes, so the exact likelihood often cannot be evaluated directly in practice. We therefore start with a $4\times 4$ lattice ($n=16$), which has $2^{16}=65{,}536$ configurations, making the computation of the exact likelihood—and hence the likelihood-based IM—feasible. Figure~\ref{fig:ising}(a) shows the observed data set, and Figure~\ref{fig:ising}(b) displays the likelihood-based IM contour for this dataset, together with the likelihood-free IM obtained using $\sum_{(i,j)\in E} z_i z_j$ as the summary statistic, Mahalanobis depth, and $M=500$. Note that the latter aligns closely with the former.

\begin{figure}[t]
\begin{center}
\subfigure[]{\scalebox{0.4}{\includegraphics{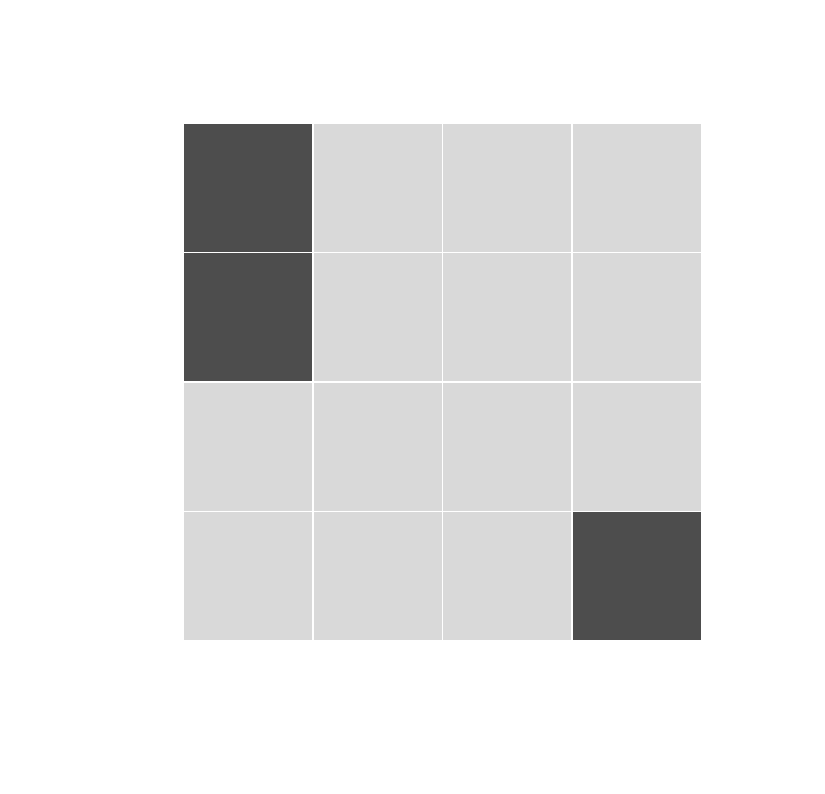}}}
\subfigure[]{\scalebox{0.4}{\includegraphics{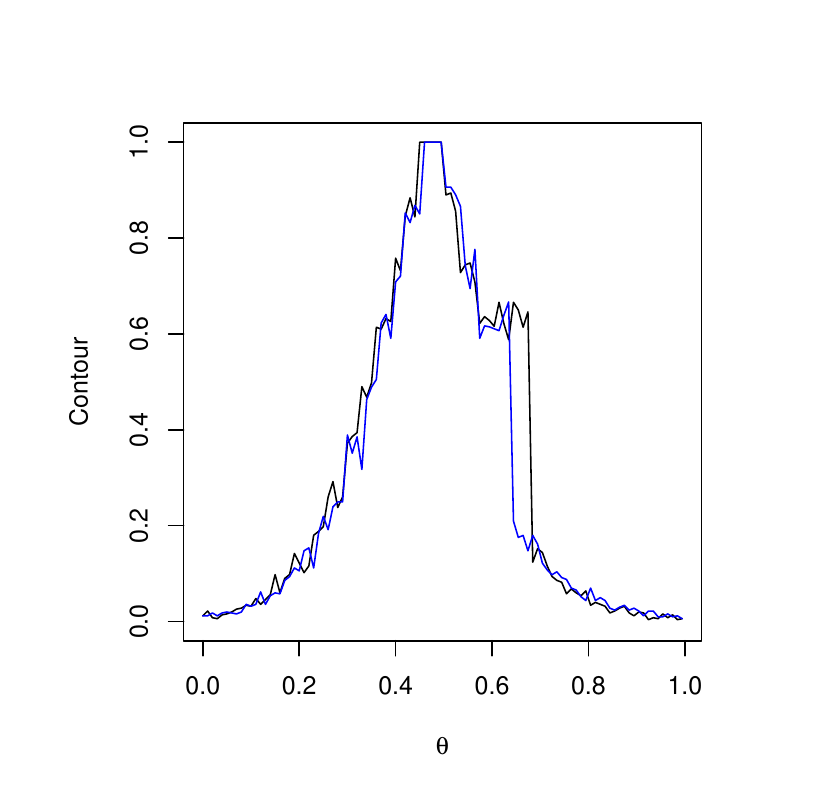}}}
\end{center}
\caption{Panel~(a): Observed data from a $4\times 4$ Ising model, where light gray cells represent $-1$ and dark gray cells represent $+1$. Panel~(b): Likelihood-based contour (black) and likelihood-free contour (blue) based on data in (a).}
\label{fig:ising}
\end{figure}

Now consider the general Ising model with $\Theta=(\beta,B)$, and, as before, let $A_n$ be the standard nearest-neighbor matrix. We take a $20\times20$ lattice ($n=400$). The likelihood-based IM is infeasible in this setting, since exact likelihood evaluation would require summing over $2^{400}$ configurations. The likelihood-free IM, however, remains feasible, as simulation from the model can be carried out using Gibbs sampling. Figure~\ref{fig:ising2}(a) shows the observed dataset, while Figure~\ref{fig:ising2}(b) displays the 0.1-level sets of two likelihood-free IMs, both constructed using $\sum_{(i,j)\in E} z_i z_j$ and $\sum_i z_i$ as summary statistics and $M=250$, one based on Mahalanobis depth and the other on Tukey depth. The dashed lines represent the true values of $\beta$ and $B$, which are captured by the 90\% confidence regions of both IMs. Note that the IM based on Mahalanobis depth is more efficient in this case, though this need not hold in general. The important point is that both constructions are valid, as confirmed in Figure~\ref{fig:ising2}(c), which shows that the empirical distribution function of the possibility contours evaluated at the true $\Theta$ is stochastically no smaller than \unif(0,1), based on 500 simulation replicates.

\begin{figure}[t]
\centering
\subfigure[]{\scalebox{0.95}{
    \includegraphics[height=4.3cm]{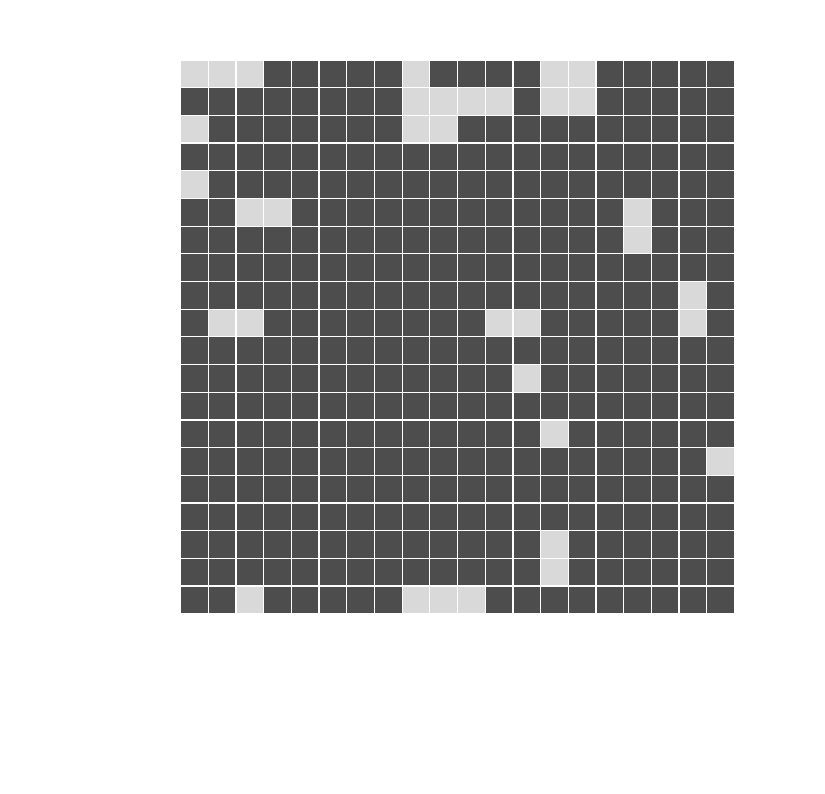}
    \label{fig:ising2b}
}}
\subfigure[]{\scalebox{0.7}{
    \includegraphics[height=5.7cm,trim=0 20 9.5cm 17cm,clip]{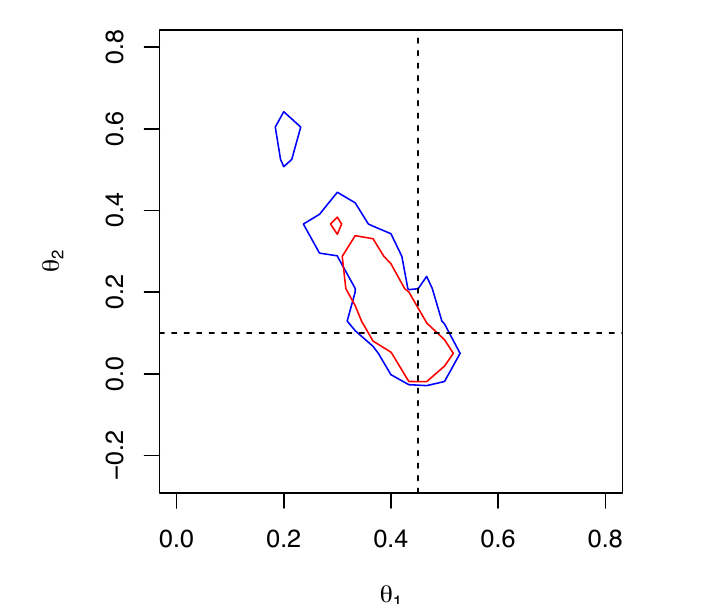}
    \label{fig:ising2b}
    }}
\hspace{0.02\textwidth}
\subfigure[]{\scalebox{0.99}{
    \includegraphics[height=4.0cm,trim= 0 5 9.5cm 17cm,clip]{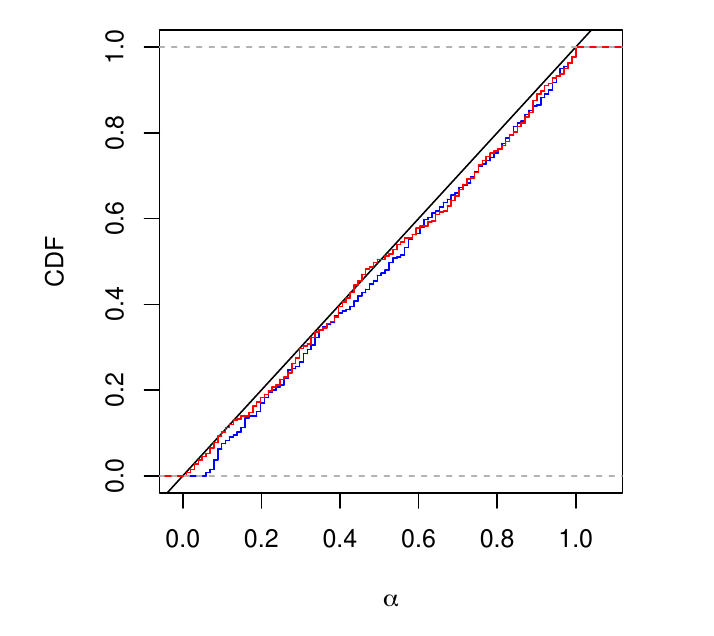}
    \label{fig:ising2c}
}}
\caption{Panel~(a): Observed data from a $20\times 20$ Ising model, where light gray cells represent $-1$ and dark gray cells represent $+1$. Panel~(b): 0.1-level sets of the likelihood-free IM contour using Tukey (blue) and Mahalanobis (red) depths. Panel~(c): Empirical distribution function of the likelihood-free IMs contours evaluated at the true $\Theta$, using Tukey (blue) and Mahalanobis (red) depths, based on 500 simulation replicates.}
\label{fig:ising2}
\end{figure}

\section{Real data example}
\label{s:real}

Measles is an extremely contagious viral infection with symptoms that include a fever, cough, runny nose, inflamed eyes and a rash. Complications from infection can be deadly, with most deaths occurring in children under the age of five \citep{Measles-Facts}. Measles was declared eliminated in the United States (U.S.) in 2000 due to the widespread availability of a safe and effective vaccine, which has been available worldwide since the 1960's, and a robust vaccination program \citep{Measles-History}. A strong resurgence of measles in 2025, however, threatens the U.S.' measles elimination status \citep{mathis2022maintenance}. The first measles death since 2015 was reported in February \citep{CNN}, and by July the number of cases had already surpassed that of the prior three decades \citep{Hopkins}. The year closed with a reported total of $2{,}281$ confirmed measles cases and three resulting deaths; the vast majority of cases occurred through local transmission, rather than importation through travellers from abroad, with 89\% of confirmed cases related to outbreaks \citep{CDC}. 

Given that almost 90\% of cases occur through local transmission, we aim to understand the static spatial configuration of measles cases across county lines in 2025 by examining how strongly neighbouring counties prefer to share the same infection status. We model measles case data published by \cite{MeaslesTracker}'s ``U.S. measles tracker'' using the Ising model of equation \eqref{eq:Ising_like}. The data are $z_i=1$ if county $i$ had at least one case in 2025 and $z_i=-1$ otherwise, $i=1, \ldots, 3235$. The symmetric coupling matrix $A_n$ has entries equal to one if two counties share a border and zero otherwise \citep[county adjacency data was downloaded from the][]{CensusBureau}. The parameters $B$ and $\beta$ respectively represent the local external pressure and the spatial transmission pressure for infection. The parameter $\beta$ is connected to the odds ratio of a county $i$ going from no measles cases ($z_i=-1$) to at least one case ($z_i=1$) given one of its neighbours $j$ goes from having no cases ($z_j=-1$) to at least one case ($z_j=1$):
\begin{align*}
\text{OR($z_i^{-1 \to 1} \mid z_j^{-1 \to 1}$)}&=\frac{\prob(z_i=1 | z_j=1, z_{-ij})/\prob(z_i=-1 | z_j=1, z_{-ij})}{\prob(z_i=1 | z_j=-1, z_{-ij})/\prob(z_i=-1 | z_j=-1, z_{-ij})} 
=\exp(4\beta),
\end{align*}
where $z_{-ij}=z \setminus \{ z_i,z_j\}$. 

Figure~\ref{fig:ising3}(a) shows the contour of the likelihood-free IM for $\Theta=(\beta,B)$ constructed again using $\sum_{(i,j)\in E} z_i z_j$ and $\sum_i z_i$ as summary statistics and $M=200$, based on the Mahalanobis depth. The central inferential question in the Ising model is whether there is spatial dependence $\beta \neq 0$, and if so, its alignment (sign) and strength. We derive a marginal IM for $\beta$, with contour shown in Figure~\ref{fig:ising3}(b). Given the natural interpretation of $\beta$ through the odds ratio, the horizontal axis in Figures~\ref{fig:ising3}(b) also displays the $\text{OR($z_i^{-1 \to 1} \mid z_j^{-1 \to 1}$)}=\exp(4\beta)$ scale. There is strong evidence supporting $\beta>0$, or equivalently $\text{OR($z_i^{-1 \to 1} \mid z_j^{-1 \to 1}$)} > 1$, indicating that neighbouring counties tend to share the same infection status. Given this apparent evidence, one may ask which claims of the form $C_\gamma : \beta > \gamma$, for $\gamma > 0$, are well supported. Although such claims are considered after inspecting the contour, the uniform calibration property of IMs established in Corollary \ref{cor:Main_LF} allows these data-dependent assessments. This question can be addressed using the marginal lower probabilities for $C_\gamma$, as shown in Figure~\ref{fig:ising3}(c). Once again, the results are also reported in the odds ratio scale. We see that the claim $\beta > 0.16$, or equivalently $\text{OR($z_i^{-1 \to 1} \mid z_j^{-1 \to 1}$)} > 1.9$, is well supported. There is evidence that the odds of a county having at least one measles case if a neighbouring county presents with at least one case are almost twice as large as if the neighbouring county has no cases, holding the other counties fixed. There is virtually no support, however, for the claim $\text{OR($z_i^{-1 \to 1} \mid z_j^{-1 \to 1}$)} > 2$. These supported claims about the magnitude of $\beta$ and $\text{OR($z_i^{-1 \to 1} \mid z_j^{-1 \to 1}$)}$ signal spatial dependence across neighbouring counties of moderate strength. While there appears to be meaningful local dependence between counties, there is no evidence of system-wide synchronization. This suggests that the spatial reach of measles outbreaks does not extend too far.

For comparison, we conduct a Bayesian analysis in which, following \citet{kim2024}, independent priors $\log \beta \sim \nm(0,1)$ and $B \sim \nm(0,1)$ are adopted and the true likelihood is replaced by a pseudo-likelihood formed from the product of the conditional distributions of each node given its neighbors. Each conditional has a simple logistic form and does not involve the global normalizing constant. The resulting (pseudo) posterior probabilities of the claims $C_\gamma : \beta > \gamma$ are also displayed in Figure~\ref{fig:ising3}(c). They were approximated using MCMC, rather than the variational Bayes approximation considered in \citealp{kim2024}. Strong posterior support is found for the more ambitious claim $\beta > 0.19$, or equivalently $\text{OR($z_i^{-1 \to 1} \mid z_j^{-1 \to 1}$)} > 2.1$. However, caution is required, since these posterior probabilities do not carry calibration guarantees.

\begin{figure}[t]
\begin{center}
\subfigure[]{\scalebox{0.37}{\includegraphics{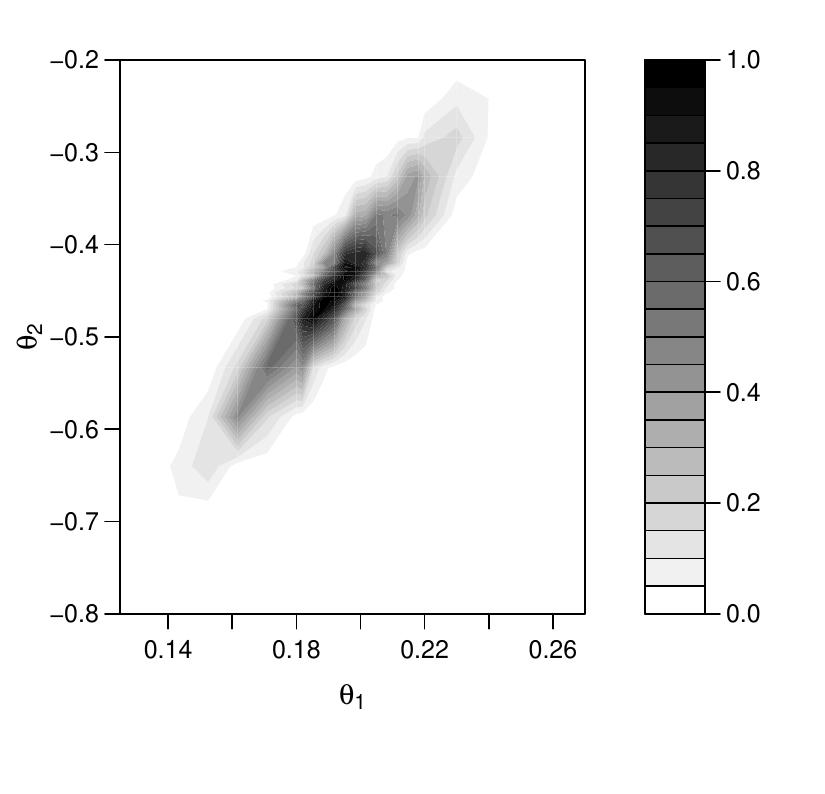}}}
\subfigure[]{\scalebox{0.37}{\includegraphics{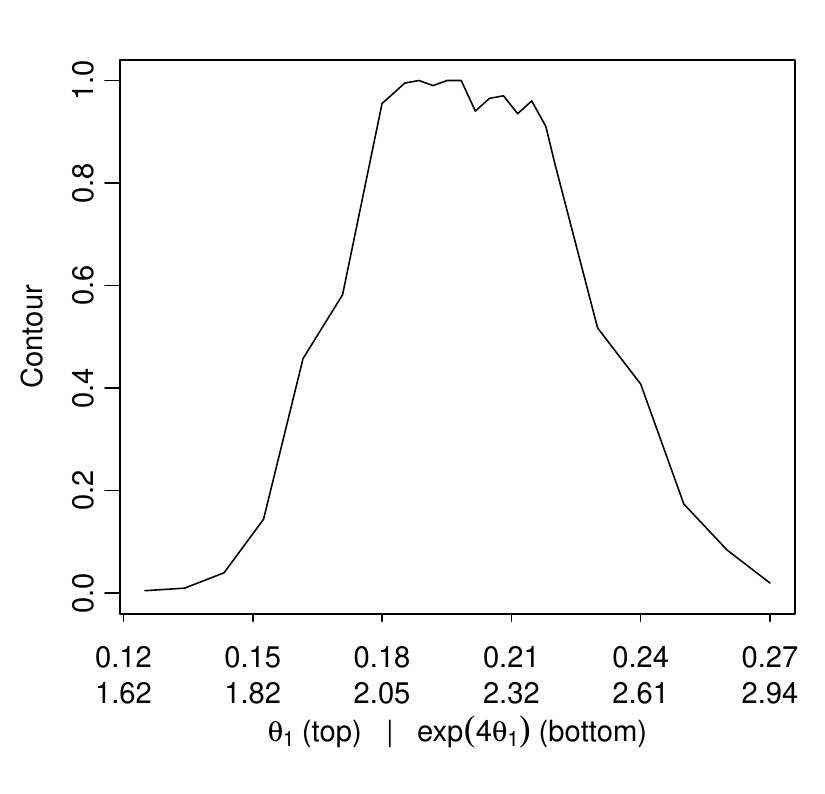}}}
\subfigure[]{\scalebox{0.37}{\includegraphics{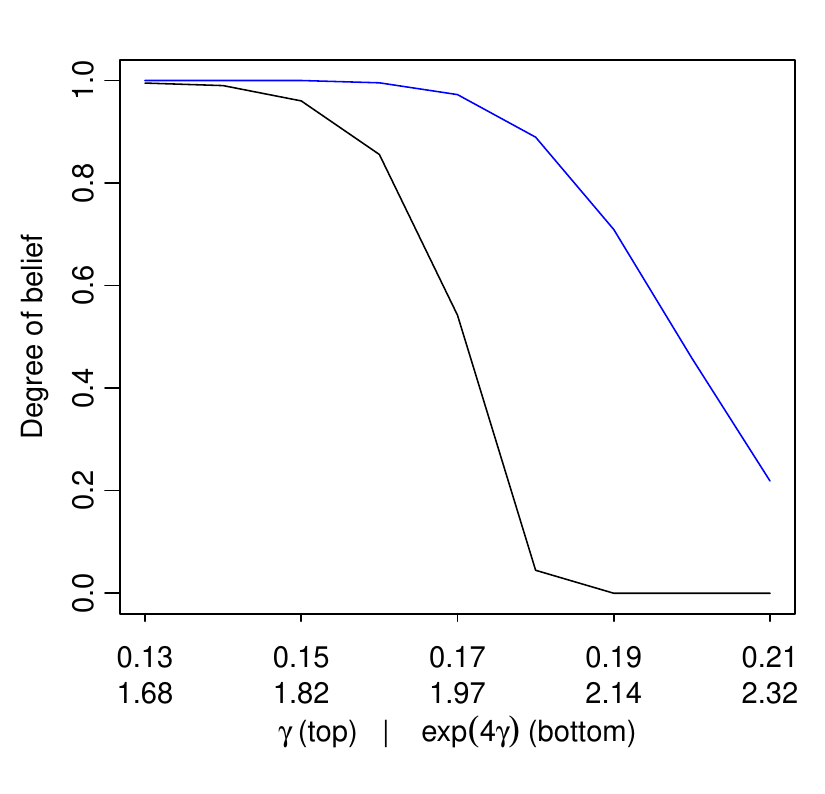}}}
\end{center}
\caption{{Panel~(a): Likelihood-free IM contour for $(\Theta_1,\Theta_2) = (\beta,B)$ based on the measles dataset. Panel~(b): Marginal IM contour for $\beta$. Panel~(c): Degrees of belief to the claims $C_\gamma : \beta > \gamma$, for $\gamma > 0$ based on the likelihood-free IM (black) and on a Bayesian analysis based on a pseudo-likelihood (blue).}}
\label{fig:ising3}
\end{figure}

\section{Conclusion}

The proposed likelihood-free IM contributes both to broad statistical methodology and to the specific development of the IM framework.
As a general statistical contribution, it provides a principled and computationally simple method for inference that bridges the gap between traditional likelihood-free frequentist and Bayesian approaches. In particular, the method admits two complementary interpretations: as a frequentist approach capable of producing finite-sample calibrated degrees of belief, and as a prior-free probabilistic approach capable of generating inferential procedures with finite-sample error rate control guarantees. From the perspective of the IM framework, the proposed construction preserves the core ranking–validification blueprint but, crucially, introduces randomness in the ranking step. Because of this, for the first time, Monte Carlo error is accounted for in the IM construction and the calibration properties of our IM inference hold exactly; this stands in contrast to the traditional IM construction in \eqref{eq:StandIMsMC} for which calibration guarantees only hold asymptotically as $L \to \infty$. This development is especially important in settings where the likelihood is expensive to compute because it is precisely in these settings where the size of $L$ is limited by computational resources. In this sense, our simulation-based, likelihood-free IM gives further momentum to what \citet{cui.hannig.im} described as ``one of the original statistical innovations of the 2010s,'' and suggests that the IM framework remains a fertile ground for developing original solutions to modern inference problems.

\bibliographystyle{apalike}
\bibliography{lit}

\end{document}